\documentclass[11pt]{amsart}

\usepackage[DIV=12,BCOR=12mm]{typearea}

\usepackage[utf8]{inputenc}
\usepackage[T1]{fontenc}
\usepackage[normalem]{ulem}
\usepackage{amssymb}
\usepackage{amsmath}
\usepackage{color}
\usepackage{arydshln}
\usepackage{tikz-cd}
\usepackage[widespace]{fourier}
\usepackage{microtype}
\usepackage{comment}
\usepackage[colorlinks,linkcolor=blue,citecolor=teal]{hyperref}
\usepackage[capitalize]{cleveref}
\usepackage{cite}
\usepackage{booktabs}
\usepackage{todonotes}
\usepackage{physics}
\usepackage{braket}
\usepackage{multirow}

\newtheorem{theorem}{Theorem}[section]

\newtheorem{corollary}[theorem]{Corollary}
\theoremstyle{definition}

\newtheorem{example}{Example}

\theoremstyle{remark}
\newtheorem{remark}[theorem]{Remark}

\newcommand{\K}{\mathbb{K}}
\newcommand{\KK}{\mathbb{K}}

\DeclareMathOperator{\SL}{SL}
\DeclareMathOperator{\GL}{GL}
\DeclareMathOperator{\PGL}{PGL}

\DeclareMathOperator{\Sp}{Sp}

\DeclareMathOperator{\diag}{diag}

\DeclareMathOperator{\stab}{stab}

\DeclareMathOperator{\Span}{span}
\newcommand{\lieg}{\mathfrak{g}}
\newcommand{\lieh}{\mathfrak{h}}
\DeclareMathOperator{\ad}{ad}

\usepackage{graphicx}
\allowdisplaybreaks[3]

\author[G. Gubbiotti]{Giorgio Gubbiotti}
\author[B. van Geemen]{Bert van Geemen}
\author[P. Vergallo]{Pierandrea Vergallo}

\address[G. Gubbiotti]{Dipartimento di Matematica ``Federigo Enriques'',  Universit\`a degli Studi di Milano, Via C. Saldini 50, 20133 Milano, Italy \& INFN Sezione di Milano, Via G. Celoria 16, 20133 Milano, Italy}
\email{giorgio.gubbiotti@unimi.it}

\address[B. van Geemen]{Dipartimento di Matematica ``Federigo Enriques'', Universit\`a degli Studi di Milano, Via C. Saldini 50, 20133 Milano, Italy}
\email{lambertus.vangeemen@unimi.it}

\address[P. Vergallo]{Dipartimento di Scienze di Base e Applicate, Università degli Studi della Basilicata, Via dell'Ateneo Lucano,  85100 Potenza, Italy \& INFN Sezione di Napoli, Via Cintia, 80126 Napoli, Italy}
\email{pierandrea.vergallo@unibas.it}

\title[]{Classification of pairs of second-order Hamiltonian operators and hydrodynamic type systems in six components}

\date{\today}

\numberwithin{equation}{section}

\begin{document}
\begin{abstract}
Following a recent characterisation of hydrodynamic systems with second-order Hamiltonian structure in terms of alternating three-forms on the projective space, we present a complete classification of systems in 6 components. We use non-trivial group action techniques to obtain 5 canonical forms. A general result for every dimension is also obtained for one of these cases, whose geometric structure is strictly related to the symplectic canonical form.  
\end{abstract}

\maketitle

\section{Introduction}

The Hamiltonian formalism of Classical Mechanics has been a cornerstone of
investigating both chaotic and regular dynamical systems since its
introduction by Sir William Rowan Hamilton in the 19th century. This
framework has not only had a profound impact on our understanding of
Classical Mechanics but also gave birth to several purely theoretical
mathematical disciplines, such as Symplectic Geometry and
Topology~\cite{CannasSilva2001,McDuffSalamonBook}. Furthermore, the
Hamiltonian formalism is fundamental to modern physics, serving as a
foundation for quantum mechanics~\cite{Sakurai1967}. For all these reasons,
Hamiltonian formalism is now contained on any basic textbook on Classical
Mechanics~\cite{LandauMech,Goldstein2002}. As noted by Arnold~\cite[Part
III]{Arnold1997}, the success of Hamilton's approach to Mechanics stems
from the fact that such an approach can be formulated in a purely geometric
form, making it more general than Lagrange's method and paving the way for
numerous extensions.

In particular, in the second half of the 20th century through the work of
Magri~\cite{Magri1976} has begun the systematic study of the Hamiltonian
structures of evolutionary PDEs. For a complete review on this topic we
refer to the review~\cite{mokhov}, here we will limit ourselves
to report the basic notion from this beautiful extension of Hamiltonian
Mechanics. So, to be more specific, a systems of evolutionary PDEs has the
following form:
\begin{equation}
    u_{t}^{i} = f^{i}(\vb{u},\vb{u}_{x},\dots,\vb{u}_{kx};x,t),
    \quad
    \vb{u} = 
    \begin{pmatrix}
        u^{1}
        \\
        \vdots
        \\
        u^{n}
    \end{pmatrix},
    \quad \vb{u}_{jx}=\frac{\partial^{j} \vb{u}}{\partial x^{j}},
    \label{eq:evol}
\end{equation}
where usually the variable $t$ has the physical meaning of time (the evolution
variable) and $x$ has the physical meaning of space (the spatial variable).
Given a (pseudo)dif\-fer\-en\-tial operator $\mathcal{P}$ we can associate to
it a \emph{bracket} defined on functionals, i.e.\ expressions of the form:
\begin{equation}
    F=\int{f(\textbf{u},\dots
    \textbf{u}_{kx};x,t)\, \dd x},
    \qquad 
    G=\int{g(\textbf{u},\dots \textbf{u}_{rx},x,t)\, \dd x},
    \label{eq:functionals}
\end{equation}
through the following formula
\begin{equation}
    \{F,G\}_\mathcal{P}=
    \int{\frac{\delta F}{\delta u^i}\mathcal{P}^{ij}\left(\frac{\delta G}{\delta
    u^j}\right)\, \dd x}
    \label{eq:bracketgen}
\end{equation}
where we make use of Einstein summation convention, and $\delta/\delta u^i$ is
the variational derivative with respect to the $i$-th field variable. The
operator $\mathcal{P}$ is called \emph{Hamiltonian} if the
bracket~\eqref{eq:bracketgen} is a Poisson bracket, that is, it is
skew-symmetric:
\begin{equation}
    \pb{G}{F}_{\mathcal{P}} 
    =
    -\pb{F}{G}_{\mathcal{P}},
    \quad \forall F,G,
    \label{eq:skew}
\end{equation}
and satisfy the Jacobi identity:
\begin{equation}
    \{\{F,G\}_{\mathcal{P}},H\}_{\mathcal{P}} 
    +\{\{H,F\}_{\mathcal{P}},G\}_{\mathcal{P}}
    +\{\{G,H\}_{\mathcal{P}},F\}_{\mathcal{P}}=0, 
    \quad \forall F,G,H.
\end{equation}
This is the analogue for PDEs of the Poisson bracket between functions
(observables in physical terminology) on the phase space.  Then, fixing a
functional $H=\int{h((\textbf{u},\dots \textbf{u}_{kx};x,t)\dd x}$, i.e.\
the \emph{Hamiltonian}, we have that the equations of motion are written
in the form:
\begin{equation}
    u^i_t = \{u^i, H\}_{\mathcal{P}}.
    \label{eq:hameq}
\end{equation}
This is the PDE equivalent of Hamilton's equations of motion, and if the
right hand side of~\eqref{eq:hameq} agrees with the right hand side
of~\eqref{eq:evol}, then we say that the PDE system is Hamiltonian.

In the sense of the previous definitions Hamiltonian operators
$\mathcal{P}$ can be (pseudo)dif\-fer\-en\-tial operators as complicated as
one wants, and in particular they can even be \emph{nonlocal}. However,
there is a relevant class of \emph{local} differential operators which can
be treated in great extent: the one of \emph{homogeneous operators}.
Indeed, one can introduce a natural grading on the space of differential
operators through the basic rules:
\begin{equation}
    \deg(u_{kx})=k,
    \quad \deg(\partial_x^\ell)=\ell.
    \label{eq:grading}
\end{equation}
Then, given a local differential operator one says that it is homogeneous
of order $m$ if all its constituent monomials have degree $m$, i.e.\ it is
of the following form:
\begin{align}\label{ops_gen}
   \begin{split}
\mathcal{P}^{ij}&=a^{ij}\!(\textbf u) \,\partial^m_x+b^{ij}_k\!(\textbf u)\,u^k_x\,\partial^{m-1}_x+\big(c^{ij}_k\!(\textbf u)\,u^k_{xx}+c^{ij}_{k\ell}(\textbf u)\,u^k_x\,u^\ell_x\big)\partial^{m-2}_x\\
&~~~
+\cdots \,+\big( r^{ij}_k\!(\textbf u)\,u^k_{mx} + \dots +r^{ij}_{k_1 \dots \,k_m}\!(\textbf u)\,u^{k_1}_x\,\cdots\, u^{k_m}_x\big) \,,
\end{split}
\end{align}  
The relevance of homogeneous operators in the theory of the Hamiltonian
structures for PDEs is due to the groundbreaking discovery of Dubrovin and
Novikov~\cite{DubNov83}, which proved that a first order homogeneous operator:
\begin{equation}
    \mathcal{A}^{ij}=g^{ij}\!(\textbf u)\,\partial_x+b^{ij}_k\!(\textbf u)\, u^k_x, 
\end{equation}
under the non-degeneracy assumption $\det g\neq 0$ is Hamiltonian if and
only if $g_{ij}=(g^{lk})^{-1}$ is a flat metric and
$\Gamma^i_{jk}=-g_{js}b^{si}_k$ are the symbols of the Levi-Civita
connection of $g$. 

Dubrovin and Novikov's result has been extended to include a broader range
of operator classes, including higher-order operators that were examined in
greater depth by Doyle in~\cite{doyle} and Pot\"emin in~\cite{potemin}.
More explicitly, they found criteria for operators of second and third order to be
classified as Hamiltonian and demonstrated the existence of a
transformation of dependent variables allowing a generic operator $\mathcal{P}$ of order $m=2,3$ to assume a
specific, simplified form called the \emph{Doyle-Pot\"{e}min form}, namely:
\begin{equation}
    \label{eq:dpgen}   
    \mathcal{P}^{ij}=\partial_x\circ \mathcal{Q}^{ij} \circ \partial_x,
\end{equation} 
where $\mathcal{Q}^{ij}$ is a homogeneous operator of order $d=m-2$. Recent developments in this direction show how this canonical
form is typical of a large number of Hamiltonian
operators~\cite{LorShaVit1} where the homogeneous operator $\mathcal{Q}$ is
of arbitrary order $d\geq0$.

Over the past few years, researchers have employed a geometric perspective
to investigate higher-order homogeneous Hamiltonian operators. This
approach is based on techniques from both differential and algebraic geometry,
with notable contributions made in references~\cite{VerVit2,FPV14,FPV16,FPV17:_system_cl,LorShaVit1,GvGV1}. Specifically,
for the non-degenerate case ($\det g \neq 0$), it has been proved
that the leading coefficient $g^{ij}$ remains unchanged under projective
transformations of the field variables and that the entire operator is
invariant under a specific related transformation involving the independent
variables $(x,t)$, called \emph{reciprocal transformations}.



Let us now go back from operators to the systems of evolutionary PDEs.  A
particularly relevant class of evolutionary systems is the one of systems of
hydrodynamic type. Those systems of PDEs are quasilinear systems of first order
evolutionary PDEs for which there exists an affinor $V_{j}^{i}$, i.e.\ a
$(1,1)$-tensor, such that:
\begin{equation}
    u^i_t=V^i_j(\textbf{u})u^j_x.
    \label{eqsys}
\end{equation}
The affinor $V^i_j$ is also known in the literature as the \emph{velocity
matrix} for the system~\eqref{eqsys}. 
If the velocity matrix $V_{j}^{i}$ can be written as the differential of a
vector field $V^{i}$, i.e.:
\begin{equation}\label{conslaws}
    V^i_j(\vb{u})=\frac{\partial V^i}{\partial u^j},
\end{equation}
then the the hydrodynamic type system is called \emph{a system of conservation
laws}. The functions $V^{i}$ are called the \emph{fluxes}.


The reasons for studying hydrodynamic type systems are various. First of all,
they are one of the simplest cases of systems of (possibly) nonlinear PDEs that
can be constructed. Second, they are related to Tsar\"ev's theory of Riemann
invariants and the generalised hodograph transformation, as discussed
in~\cite{tsa2}. Finally, their relationship with Hamiltonian structures is
well known~\cite{tsa1}, see also~\cite{sevennec}, and the equations of
motion~\eqref{eq:hameq} are particularly easy to express:
\begin{equation}
    u^i_t = V^i_j(\textbf{u})u^j_x
    = \mathcal{P}^{ij}\left(\frac{\delta H}{\delta u^j}\right).
\end{equation}
In the literature, the Hamiltonian structures of hydrodynamic type systems
are very well studied, see for instance~\cite{FPV14,VerVit2,GvGV1}.
In particular, we mention that in the case of second- and third-order
operators, i.e., references~\cite{FPV14} and~\cite{VerVit2}, respectively,
the invariance properties of the leading order term $g^{ij}$ allow for the
possibility of a characterisation of the associated systems of equations of
hydrodynamic type. Moreover, based on the results of~\cite{VerVit2}, in our
previous work~\cite{GvGV1} we proved that a second-order homogeneous
Hamiltonian operator and its associated system of hydrodynamic type in $n$
components can be seen as the coefficients of a single object: a totally
skew-symmetric three-form, and we used this knowledge to classify the pairs
of operator-system up to $n=4$.

\

In this paper, we present a classification of hydrodynamic-type systems
admitting a second-order homogeneous Hamiltonian structure in six components. The tools employed here are
primarily drawn from algebraic geometry and the theory of group actions.
This approach is consistent and extends the results in the literature,
in particular it extends our previous result~\cite{GvGV1}. To be more precise
the plan of the paper is the following: in Section~\ref{sec:back} we state
the needed known results on second order homogeneous Hamiltonian operators
and their associated hydrodynamic type systems. In Section~\ref{sec:class}
we present the full classification of the pairs operator-systems in six
components. Finally, in Section~\ref{sec:concl} we give some conclusions
and outlook on possible future results.

\section{Background material}
\label{sec:back}

Following the previous works~\cite{VerVit1,VerVit2,GvGV1}, in this section we recall the needed results on the projective geometric properties of
second-order homogeneous Hamiltonian operators and their associated systems of
conservation laws.
In particular, we recall the classification procedure of the pairs of
operators-systems as devised in~\cite{GvGV1}.

\subsection{Characterisations of pairs}

The general structure of a second-order homogeneous Hamiltonian operator is:
\begin{equation}
    \label{1}
    \mathcal{P}^{ij}=g^{ij}\partial_x^2+b^{ij}_ku^k_x\partial_x+c^{ij}_ku^{k}_{xx}
    +c^{ij}_{kh}u^k_xu^h_x,
\end{equation}
where the coefficients $g^{ij},b^{ij}_k,c^{ij}_k$ and $c^{ij}_{kl}$ depend only
on the field variables. We say that the operator is non-degenerate if the
co-metric matrix $g=(g^{ij})_{i,j=1}^{n}$ has $\det g\neq 0$. Throughout the
rest of the paper we will only deal with non-degenerate operators.

As recalled in the Introduction the second-order second-order homogeneous
Hamiltonian operators can be put through some local change of coordinates into
the Doyle--Pot\"emin form~\eqref{eq:dpgen}, which in this particular case
assumes the following structure:
\begin{equation}
    \mathcal{P}^{ij}=\partial_x\circ g^{ij}\circ \partial_x, 
\end{equation}
where we take $\mathcal{Q}^{ij}=g^{ij}$ in the canonical form.
Moreover, in this coordinates the inverse of the leading coefficient reads as
\begin{equation}\label{3}
    g_{ij}=T_{ijk}u^k+g^0_{ij},
\end{equation}
where $T=T_{ijk}\dd u^i\wedge \dd u^j\wedge \dd u^k$, $g^0=g^0_{ij}\dd
u^i\wedge \dd u^j$ are respectively an alternating three-form and an
alternating two-form with constant coefficients, i.e.\ $T\in
\Lambda^{3}\KK^{n}$ and $g^{0}\in\Lambda^{2}\KK^{n}$, with
$\mathbb{K}=\mathbb{R},\mathbb{C}$.

Let us consider the $n$-dimensional projective space
$\mathbb{P}^n=\mathbb{P}(\mathbb{K}^{n+1})$, and let $u^1, \ldots ,u^n,u^{n+1}$
be the coordinates on $\K^{n+1}$. Following~\cite{VerVit2}, we can define a
homogeneous version $G$ of $g$ in these coordinates, such that
equation~\eqref{3} becomes
\begin{equation}
    G_{ij}=T_{ijk}u^k+g^0_{ij}u^{n+1}.
    \label{eq:Gdef}
\end{equation}
Let us observe that $G$ as defined in~\eqref{eq:Gdef} is \emph{not} an
alternating two-form. However, we can associate to $G_{ij}$ an alternating
three-form $\tilde{T}\in\Lambda^{3}\KK^{n+1}$ defined via its components as as
follows:
\begin{equation}
    \label{5}
    \tilde{T}_{ijk}=
    \begin{cases}T_{ijk} \qquad i,j,k\neq n+1,\\
    +g^0_{ij}\qquad k=  n+1, \\
    -g^0_{ik}\qquad j=n+1,\\
    +g^0_{jk}\qquad i=n+1.
    \end{cases}
\end{equation}
The construction we just explained yields the following result:
 
\begin{theorem}[\cite{VerVit2}]\label{th:corresp}
  There is a bijective correspondence between the leading coefficients of
  second order homogeneous Hamiltonian operators in Doyle-Pot\"{e}min form and
  the three-forms $\tilde{T}$. Moreover, the bijective correspondence is
  preserved by projective reciprocal transformations up to a conformal factor.
\end{theorem}

That is, the above construction completely characterises the  second-order
homogeneous Hamiltonian operators in terms of three-forms, and up
to a ``large'' group of transformations. 

Moreover, in~\cite{VerVit1} the authors gave a set of necessary conditions for
a system of hydrodynamic type~\eqref{eqsys} to admit a Hamiltonian formulation
with a homogeneous second-order operator.  This is easily explained introducing
potential coordinates, i.e.\ $b^i_x=u^i$. Indeed, in such coordinates the
operator~\eqref{1} and the system given by~\eqref{conslaws} take the following
simple forms respectively:
\begin{equation}
    \label{finop}
    \mathcal{P}^{ij}=-g^{ij}, \qquad b^i_t=V^i(b_x).
\end{equation}
Then, for a system of conservation laws, the compatibility conditions to be
Hamiltonian with a second order Hamiltonian operator~\eqref{finop} are
expressed by the following theorem:

\begin{theorem}[\cite{VerVit1, VerVit2}] 
    The necessary conditions for a second-order homogeneous Hamiltonian
    operator $\mathcal{P}$~\eqref{finop} to be a Hamiltonian operator for a
    quasilinear system of first-order conservation laws are
    \begin{subequations}\label{eq:37}
      \begin{gather}
        \label{eq:451}
        g_{qj}V^j_{,p} + g_{pj}V^j_{,q} = 0,
        \\
        \label{eq:38}
        g_{qk}V^k_{,pl} + g_{pq,k}V^k_{,l} + g_{qk,l}V^{k}_{,p}= 0.
      \end{gather}
    \end{subequations}
\end{theorem}

Remarkably enough, in this particular case the system~\eqref{eq:37} can be
explicitly solved and we obtain that for a fixed non degenerate $g_{ij}$ the
fluxes are of the following form
\begin{equation}\label{cond0}
    V^i=g^{is}\left(A_{sl}u^l+B_s\right),
\end{equation}
where $A=A_{ij} \dd u^i\wedge \dd u^j$ is an alternating two-form and $B=B_i \dd u^i$ is a
one-form with constant coefficients, i.e.\ $A\in \Lambda^{2}\KK^{n}$ and
$B\in\Lambda^{1}\KK^{n}\cong\KK^{n}$.

If $g$ is the leading coefficient of a second-order Hamiltonian operator for a
system given by the fluxes $V^i$, we say that $(g,V)$ forms a \emph{compatible
pair}.  We denote the space of the pairs operator-system in $n$ components by
$\mathcal{Y}_{n}$.

Let us now consider~\eqref{cond0}, written in the following equivalent form:
\begin{equation}
    g_{is}V^s=A_{il}u^l+B_i.
\end{equation}
In projective coordinates for $g$, this latter equality can be rewritten as:
\begin{equation}\label{pluck}
    T_{ijs}u^jV^s+g^0_{is}V^s=A_{is}u^s+B_i, \qquad i=1,2,\dots , n.
\end{equation}
In our previous paper~\cite{GvGV1}, we proved that there exists a three-form
$\omega\in\Lambda^3\KK^{n+2}$ such that equation~\eqref{pluck} is the
annihilation set of lines for the form, i.e. those lines whose pullback with respect to the form vanishes,
and this is in bijection with the form
itself. That is, the following Theorem holds true:

\begin{theorem}[\cite{GvGV1}]\label{thm11}
    There exists a correspondence between the pair $(\mathcal{P},V)$ of the
    second-order operator and the associated systems in $n$ components and
    alternating three-forms in $n+2$ components.  Explicitly, there exists a
    bijective map $\Phi\colon \Lambda^3\mathbb{K}^{n+2}\rightarrow
    \mathcal{Y}_{n}$ defined as:
    \begin{align}
       \Lambda^3\mathbb{K}^{n+2}\ni ( \omega_{ijk}) 
       \mapsto \left(\omega_{ijk}u^k+\omega_{ij\,n+1}, g^{is}(\omega_{ij\, n+2}u^j+\omega_{i\, n+1\, n+2})\, \right)
       \in \mathcal{Y}_{n},
    \end{align}
    with inverse 
    $\Phi^{-1}\colon\mathcal{Y}_{n} \rightarrow \Lambda^3\mathbb{K}^{n+2}$ defined
    as:
    \begin{align}
        \mathcal{Y}_{n}\ni(\mathcal{P},V)
        \mapsto
        \Omega=
        \tilde{T}_{ijk} + A\wedge \dd u^{n+2}
        +B \wedge \dd u^{n+1}\wedge \dd u^{n+2}\in\Lambda^{3}\K^{n+2},
    \end{align}
    where $\tilde{T}\in\Lambda^{3}\K^{n+1}$ is defined in equation~\eqref{5},
    and $A\in\Lambda^{2}\K^{n}$, $B\in\Lambda^{1}\K^{n}\cong\K^{n}$ are the
    constants appearing in equation~\eqref{pluck}.
\end{theorem}

Intuitively, \Cref{thm11} means that we can consider a second-order operator and
the associated systems in $n$ components as part of a single object in $n+2$
components. This result, together with the known invariance properties of
congruence lines~\cite{agafer1,agafer2} reduces the problem of classifying all pairs
$(\mathcal{P},V)$, to the problem of classifying alternating three-forms in
$n+1$ variables and produce a list of canonical alternating two-forms with respect
of their stabiliser groups.

\subsection{Outline of the classification approach}

In this subsection we specify the strategy to classify the pairs
$(\mathcal{P},V)$ as we outlined it previously in~\cite{GvGV1}.
Consider the standard vector space of dimension $n+2$ over the field
$\KK$, namely $\KK^{n+2}$ and variables $u^1,\ldots,u^{n+2}$. Then, 
the space of alternating $k$-forms $\Lambda^k\KK^{n+2}$ is such that:
\begin{equation}
    \Lambda^k\KK^{n+2}=
    \Span_{\KK}
    \Set{\dd u^{i_1}\wedge \cdots \wedge \dd u^{i_k} | 1\leq i_1 < i_2 <\ldots<i_k\leq n+2}.
\end{equation}
In particular, we have the following decompositions of 
$\Lambda^3\KK^{n+2}$~\cite[Eq. (3.6-7)]{GvGV1}:
\begin{equation}
    \begin{aligned}
    \Lambda^3\K^{n+2}
    &=
    \Lambda^3\K^{n+1}\oplus\Lambda^2\K^{n}\oplus\Lambda^1\K^{n}
    \\
    &=\Lambda^3\K^{n+1}\oplus\Lambda^2\K^{n+1}.
    \end{aligned}
    \label{eq:decompNppN}
\end{equation}
More explicitly given $\Omega\in  \Lambda^3\mathbb{K}^{n+2}$ we can
write it as follows:
\begin{equation}
\label{dec1}
    \begin{aligned}
    \Omega &= \tilde{T}\,+\,A\wedge \dd u^{n+2}\,+\,B\wedge \dd u^{n+1}\wedge \dd u^{n+2}
    \\
    &=\tilde{T}\,+\,\tilde{A}\wedge \dd u^{n+2}~,    
    \end{aligned}
\end{equation}
where
$\tilde{T}\in\Lambda^3\mathbb{K}^{n+1}$, $A\in \Lambda^2\mathbb{K}^{n}$,
$B\in\Lambda^1\mathbb{K}^n\cong \K^n$ and $\tilde{A}\in \Lambda^2\mathbb{K}^{n+1}$.
In~\eqref{dec1} the first (second) line corresponds to the first
(second) line of~\eqref{eq:decompNppN}. The three-form $\tilde{T}$
corresponds to the leading coefficient $g$, while the two-form $\tilde{A}$ corresponds
to the system of conservation laws.

So, to classify the pairs $(\mathcal{P},V)$ we need to classify the
three-forms on $\mathbb{K}^{n+2}$ decomposed as in equation \eqref{dec1} up
to the action of the projective linear group $\PGL(n+1,\mathbb{K})$ in $n+1$ dimensions.
Moreover, we need to add the two following consistency conditions:
\begin{enumerate}
    \item the alternating three-form $\tilde{T}\in\Lambda^3\mathbb{K}^{n+1}$ must be 
        non-degenerate, i.e. it must define a non-degenerate second-order 
        homogeneous Hamiltonian operator;
    \item the alternating two-form $\tilde{A}\in\Lambda^2\mathbb{K}^{n+1}$ must be
        non-null, i.e. it must define a non-trivial system.
\end{enumerate}

Point (1) implies that we can use the classification of second-order operators
obtained in~\cite{VerVit2}. That is, we can start with a fixed alternating
three-form in $\tilde{T}\in \Lambda^3\mathbb{K}^{n+1}$ in a standard form,
and act with transformations leaving it invariant, i.e.\ with elements
of its stabiliser group $M\in\stab(\tilde{T})$ on the two form $\tilde{A}$
to find a standard form for it.

In~\cite{GvGV1}, we carried out the previous program for $n=2,4$, while in
this paper we carry it out for $n=6$ using the standard forms of the operators
given in~\cite{VerVit2}. Since, as mentioned above we need the stabiliser groups,
we will refer to~\cite{CohenHelminck1988} where those groups were given explicitly,
for more information see Appendix~\ref{app:forms}. 


\section{Classification of pairs $(\mathcal{P},V)$ for $n=6$}
\label{sec:class}

In this section, we present the classification of pairs $(\mathcal{P},V)$
for $n=6$ using the technique explained at the end of the previous
section.  As already mentioned, to be consistent with the existing
literature we will use the explicit projective classification
of second-order homogeneous Hamiltonian operators under the non-degeneracy
assumption of the leading coefficient $g^{ij}$ given in \cite{VerVit2}.
This consists of five different equivalence classes: 
\begin{enumerate}
    \item $\tilde T^{\text{I}}=(\dd u^1 \wedge \dd u^2 +\dd u^3 \wedge \dd u^4 + \dd u^5 \wedge \dd u^6)\wedge 
    \dd u^{7};$
    \item $\tilde T^{\text{II}}=du^1\wedge du^2 \wedge du^3 +du^4\wedge du^5 \wedge du^6 +du^1 \wedge du^4 \wedge du^7;$
    \item $\tilde T^{\text{III}}=\dd u^4\wedge \dd u^5\wedge \dd u^6 
    + \dd u^7\wedge(\dd u^1\wedge \dd u^4+\dd u^2\wedge \dd u^5+\dd u^3\wedge \dd u^6);$
    \item $\tilde T^{\text{IV}}=\dd u_1 \wedge \dd u_2 \wedge \dd u_3 
    + \dd u_4 \wedge \dd u_5 \wedge \dd u_6 
    + (\dd u_1 \wedge \dd u_4 +\dd u_2 \wedge \dd u_5) \wedge \dd u_7,$
    \item $\tilde T^{\text{V}}=du^1 \wedge du^2 \wedge du^3 +du^4 \wedge du^5 \wedge du^6
+du^7 \wedge (du^1 \wedge du^4 +du^2 \wedge du^5 +du^3 \wedge du^6).$
\end{enumerate}
Corresponding to the above five 
orbits one can compute the associated second-order operators as in \cite{VerVit2}. 
We here classify for each operator the compatible quasilinear systems of 
first order PDEs, in terms of $g^{ij},A_{ij}$ and $B_i$ and referring to \eqref{cond0} for the explicit construction of the fluxes $V^i$. 

Before presenting the results, we observe that the first orbit is treated
in a more general $n+2$ dimensional case, as it is observed that this case and
the special $n=6$ case pose no additional difficulty whatsoever. This is because
this is the only case where the leading coefficient $g$ constant,
and the stabiliser of the corresponding form $\tilde{T}$ is related to 
the symplectic group $\Sp(n+2,\KK)$.

\subsection{First orbit: the symplectic case}

Following~\cite{VerVit2} the first leading coefficients of the operators $\mathcal{C}$ 
is inverse matrix of the following:
\begin{equation}g_{ij}^1=\begin{pmatrix}
0&0&0&1&0&0\\
0&0&0&0&1&0\\
0&0&0&0&0&1\\
-1&0&0&0&0&0\\
0&-1&0&0&0&0\\
0&0&-1&0&0&0
\end{pmatrix}
\end{equation}
we have $\det (g^1_{ij})=1$, and whose associated three form is:
\begin{equation}
    \tilde{T}_1=(\dd u^1 \wedge \dd u^4 +\dd u^2 \wedge \dd u^5 + \dd u^3 \wedge \dd u^6)\wedge 
    \dd u^{7}.
    \label{eq:T1}
\end{equation}
This form corresponds, up to an inessential permutation of variables,
to the form $f_8$ of~\cite{CohenHelminck1988}, see also Appendix~\ref{app:forms}.

It is easy to see that the three-form~\eqref{eq:T1} can be written as:
\begin{equation}
    \tilde{T}_1=\eta_3\wedge \dd u^{7},
\end{equation}
where
\begin{equation}
    \eta_3 = \dd u^1 \wedge \dd u^2 +\dd u^3 \wedge \dd u^4 + \dd u^5 \wedge \dd u^6,
\end{equation}
is the standard symplectic form of $\KK^6$. This is clearly a generalisation
of what we considered in~\cite[\S 5.b]{GvGV1}, where considering one of
the two non-degenerate orbits for $n=4$ we observed it was expressible
as $\eta_2\wedge\dd u^5$, with $\eta_2$ the standard symplectic form of $\KK^4$.
Observe that $\stab (\tilde{T}_1)=\Sp(6)\rtimes \KK^6$, see Appendix~\ref{app:forms}.

Now, we should use the stabiliser to act on the two-form $\tilde{A}$, or
equivalently on a $(n+1)\times (n+1)$ skew-symmetric matrix. The action of
the symplectic group on skew-symmetric matrices is a well studied topic,
see e.g.~\cite{Thompson,Tyurin,Rodman,Lancaster,Dmytryshyn_KS,Dmytryshyn_D}.  For
this reason, it is possible to tackle the problem in full generality, i.e.
to consider the following alternating three-form in
$\Lambda^3\mathbb{K}^{2n+1}$:
\begin{equation}
    \tilde{T}_1^{2n+1}=\eta_{n}\wedge \dd u^{2n+1},
\end{equation}
where $\eta_n$ is the standard symplectic form of $\KK^{2n}$:
\begin{equation}
    \eta_n = \sum_{i=1}^{n} \dd u^{2i-1}\wedge\dd u^{2i}.
\end{equation}
The case $n=3$ will be then obtained as a particular one, and also the cases 
$n=2,4$ discussed in~\cite{GvGV1} are re-obtained. 

\begin{remark}
    We observe that, in principle, one could consider instead of $\eta_n$ 
    a generic non-degenerate two-form $\omega\in\Lambda^2\mathbb{K}^{2n}$.
    Indeed, applying the Darboux theorem to $\omega$, we can map it into $\eta$,
    so our choice is not restrictive.    
\end{remark}

Using the construction introduced in \cite{VerVit2}, we can associate to $\tilde{T}_1$ 
the second-order homogeneous operator in constant form (and Doyle-Pot\"emin canonical form):
\begin{equation}\label{cpontham}
    \mathcal{C}^{ij}=\partial_x\circ g^{ij}\circ \partial_x=g^{ij}\partial_x^2, \qquad g^{ij}=\begin{pmatrix}
        0&-\mathbb{1}\\\mathbb{1}&0
    \end{pmatrix}
\end{equation}
with $\mathbb{1}$ the identity $n\times n$ matrix.  
We now want to characterise the Hamiltonian system of first-order PDEs, 
associated to \eqref{eqsys}:
\begin{equation}\label{sysgen}
    V^i=g^{is}\left(A_{sl}u^l+B_s\right) \quad \Rightarrow \quad V^i_{,j}=g^{is}A_{sj}, 
\end{equation}
so that $B$ does not play any role in this context. The resulting system 
is then always linear and not of great interest from a physical point of 
view. Moreover, in this case one can always assume $\det(A)\neq 0$, otherwise
one could choose a particular change of variables $\{\tilde{u}^1,\dots \tilde{u}^{2n}\}$ such that 
\begin{equation}
    A_{i\, 2n-1}=A_{2n-1\, i}=A_{i\, 2n}=A_{2n\, i}=0, \qquad i=1,2,\dots 2n
\end{equation}
so that the associated system $V^i_{,j}$ becomes a $(n-2) \times (n-2) $ 
hydrodynamic type system, that is the system is degenerate and not of interest.

If we want to classify such Hamiltonian systems we can apply Theorem \ref{thm11}, 
so that we need to act on $A$ with the stabilizer of $\tilde{T}_1$. Note, in addition, 
that $A\in\Lambda^2\mathbb{K}^{2n}$ is a symplectic form (not a priori in 
Darboux coordinates as $\eta$).
In this case we have $\stab(\tilde{T}_1) = \Sp(2n)\rtimes \K^{2n}$ which has the following
matrix representation:
\begin{equation}
    \stab (\tilde{T}_1) = \Set{
    M \in \SL(2n+1,\K) |
    M = 
    \begin{pmatrix}
        C & 0
        \\
        x^T & 1
    \end{pmatrix},
    \qquad
    C\in\Sp(2n),x\in\K^{2n}}.
\end{equation}

As observed it is sufficient to classify only $A\in \Lambda^2\mathbb{K}^{2n}$ and by using the map
\begin{equation}
    \varphi_\eta: \Lambda^2 V\rightarrow \mathbb{K}, \quad v_i\wedge v_j\mapsto \eta(v_i,v_j)
\end{equation}we obtain that being $\eta$ non-degenerate the maps induces the splitting given by its kernel and the image subspace
\begin{equation}\label{decom_eta}
    \Lambda^2\mathbb{K}^{2n}=\mathbb{K}\eta \oplus \Theta,
\end{equation}
where $\Theta = \Set{\alpha\in\Lambda^2\K^{2n} | \alpha\wedge \eta^{\wedge^{n-1}}=0}$.
We can finally remark that 
\begin{equation}
    \Theta \cong \Set{K\in \text{Alt}(2n,\K) | \sum_{i=1}^{n}K_{i,i+n}=0},
    \label{eq:thetaisom}
\end{equation}
where $\text{Alt}(2n,\K)$ indicates the alternating $2n\times 2n$ matrices with entries in $\K$.

We make use of the following characterisation on the simultaneous
diagonalisation of alternating matrices:

\begin{theorem}[\cite{Lancaster}]\label{thm0}
        Let $L$ and $M$ be alternating $2n\times 2n$ matrices. Then, there exists $g\in GL_{2n}(\mathbb{K})$ such that 
        \begin{equation}
            gLg^T=\begin{pmatrix}
                0&\mathbb{1}\\-\mathbb{1}&0
            \end{pmatrix} 
            \qquad 
            \text{and} 
            \qquad 
            gMg^T=\begin{pmatrix}
                0&D\\-D^{T}&0
            \end{pmatrix}
        \end{equation}
        where $D$ is in one of the two following form
        \begin{equation}\label{cases}
            \begin{pmatrix}
                \lambda_1&&&\\
                &\lambda_2&&\\
                &&\ddots&\\
                &&&\lambda_n
            \end{pmatrix}, \qquad \text{or} \qquad \begin{pmatrix}
                J_1&&&\\
                &J_2&&\\
                &&\ddots&\\
                &&&J_k
            \end{pmatrix}
        \end{equation}
        and $J_\ell$ is a Jordan block with eigenvalue $\lambda_\ell$, such that $\sum |J_\ell|=n$.
\end{theorem}

So, we have the following result:
\begin{theorem}\label{thm1}
    A hydrodynamic type system of conservation laws 
    \begin{equation}
        u^i_t=(V^i)_x, \qquad i=1,2,\dots n,
    \end{equation}
    that is Hamiltonian with a second-order operator \eqref{cpontham} 
    can always be mapped by reciprocal-projective transformations into 
    one depending at most on $n$ parameters. 
\end{theorem}

\begin{proof}
    Following the decomposition in~\eqref{decom_eta} we have that
    a generic alternating two-form $A\in\Lambda^2\K^{2n}$ can be written
    as $A=\alpha \eta + \theta$, with $\alpha\in\K$ and $\theta\in\Theta$.
    Then we apply Theorem~\ref{thm0} to the matrices associated to the 
    forms $\eta$ and $\theta$, say $L$ and $M$ respectively. In particular, 
    we note that $L$ is already written in Darboux form. So, $g$ is a symplectic
    matrix, and the symplectic form $M$ can be of put in one of the two forms 
    described in \eqref{cases}. The free parameters involved are then the scalar 
    $\alpha$ and the eigenvalues
    of the matrix $D$ (counted with their multiplicities), 
    $\Set{\lambda_1,\ldots,\lambda_n}$, the latter being subject to the 
    ``trace'' condition~\eqref{eq:thetaisom}, i.e. $\sum_{i=1}^{n}\lambda_i=0$. 
    This implies that the number of free parameters is at most $1+n-1=n$ concluding
    the proof.
\end{proof}

Theorem~\ref{thm1} applied to \eqref{sysgen} gives us the form of
the system:
\begin{equation}
    V^i_{,j}=\begin{pmatrix}
        0&-\mathbb{1}\\
        \mathbb{1}&0
    \end{pmatrix}\begin{pmatrix}
        0&\alpha I_n+D\\-(\alpha I_n+D)&0
    \end{pmatrix}=\begin{pmatrix}
        \alpha I_n+D&0\\0&\alpha I_n+D
    \end{pmatrix}
\end{equation}
with the additional zero-trace condition of the matrix $D$, i.e.\ $D\in\mathfrak{sl}_n(\K)$.
So, as expected, the resulting system is linear, implying that the class of system 
in this orbit is linearisable. In particular, we have the following corollary:

\begin{corollary}\label{cor1} 
    Hydrodynamic type systems~\eqref{eqsys}
    which are Hamiltonian with \eqref{cpontham} are in block-diagonal form
    where each block is repeated twice. 
\end{corollary}

This last result directly implies that all the eigenvalues are double, 
as firstly proved in~\cite[Proposition 20]{VerVit2}. In particular, if the Jordan-block form of $D$ admits $n$ eigenvalues (i.e. $D$ is diagonal) we re-obtain also the second statement of Proposition 20. However, a non-diagonal structure of $D$ implies that the resulting matrix $V^i_j$ is non-diagonalisable.

We can finally reconstruct as particular cases the following examples for $n=1,2$, also presented in \cite{GvGV1}.
\begin{example}
    For $n=1$, every $2\times 2$ system with the above property 
    depend on one arbitrary parameter and it is always diagonal. 
    In this case, the free parameter can be scaled 
    away using an additional transformation.
\end{example}

\begin{example}
    For $n=2$, we obtain that $4\times 4$ systems depend at most on two
    parameters. Moreover, the structure of $D$ in $A$ can be one of the
    following:
    \begin{equation}
        D_1=\begin{pmatrix}
            \lambda_1&0\\
            0&\lambda_2
        \end{pmatrix}\qquad \text{or} \qquad D_2=
    \begin{pmatrix}
            \lambda_1&1\\
            0&\lambda_1
        \end{pmatrix}
    \end{equation}
    for arbitrary $\lambda_1$ and $\lambda_2$, with the additional
    requirement to be traceless. However, by using Corollary \ref{cor1} we
    obtain that the system associated to $D_1$ is diagonal whereas the one
    associated to $D_2$ is non-diagonalisable.  We finally stress that in
    $D_2$ the only possible case is that $\lambda_1=-\lambda_1=0$, so that
    the block is degenerate. To compare, we finally stress that in
    \cite{GvGV1}, the degenerate case was neglected.
\end{example}

    The last example in particular shows that not all systems associated to a homogeneous second-order Hamiltonian operator are diagonalisable. However, the additional condition $A_{13}+A_{24}=0$ restricts the possibility only to $A_1$, indeed otherwise $2\lambda_1=0$ implies the degeneracy of $A_2$. We stress that for $A_1$ we have $\lambda_2=-\lambda_1$ and the general structure of the two-form $A$ for $n=2$ is
    \begin{equation}
        A=\alpha \eta + \lambda_1(du^1\wedge du^3-du^2\wedge du^4).
    \end{equation}



\subsection{Second orbit}
In this paragraph, we consider $\tilde T^{\text{II}}$.
The associated leading coefficient is
\begin{equation}
g_{ij}^{\text{II}}=\begin{pmatrix}
0&u^3&-u^2&1&0&0\\
-u^3&0&u^1&0&0&0\\
u^2&-u^1&0&0&0&0\\
-1&0&0&0&u^6&-u^5\\
0&0&0&-u^6&0&u^4\\
0&0&0&u^5&-u^4&0
\end{pmatrix}\, .
\end{equation}

We stress that the representative here used corresponds to $f_5$ in \cite{CohenHelminck1988}, that is $\tilde{T}^{\text{II}}=f_5$ (see Appendix ~\ref{app:forms}). For this reason, we can use ~\cite{CohenHelminck1988} to construct the stabiliser
$\stab (f_5)$ as composed of matrices with the following 
shape\footnote{Here we correct a typo in the formula for the matrices
of the group $G_5:=\stab(f_5)$ presented in~\cite[\S 3.5]{CohenHelminck1988}: therein the coefficient $(5,1)$ of the matrix $M_5$ has the incorrect sign.}:
\begin{equation}
    M_5 = 
    \begin{pmatrix}
        (ad-bc)^{-1} & 0 & 0 & 0 & 0 & 0 & 0 
        \\
        \alpha_2 & a & b & -\lambda_1 & 0 & 0 & 0
        \\
        \alpha_1 & c & d & \lambda_2 & 0 & 0 & 0
        \\
        0 & 0 & 0 & (eh-fg)^{-1} & 0 & 0 & 0
        \\
        \lambda_3 & 0 & 0 & \alpha_4 & e & f & 0
        \\
        -\lambda_4 & 0 & 0 & \alpha_3 & g & h & 0
        \\
        \alpha_5 & \lambda_2 & \lambda_1 & \alpha_6 & \lambda_4 & \lambda_3 & (ad-bc) (eh-fg)
    \end{pmatrix}
\end{equation}
where $\alpha_i,\lambda_i,a,b,c,d,e,f,g,h\in\K$, subject to the condition that
$(ad-bc) (eh-fg)\neq 0$. We do not consider a discrete part of the group, since it will
not be used in our discussion. That is, we have $\dim\stab(f_5)=18$, and
it can be seen that the group has structure $\stab(f_5) \cong \K^{10}\rtimes (\GL(2,\K)\times\GL(2,\K))$. 

To find a standard form for the the generic two-forms 
$\tilde{A}^{\text{II}}=\tilde{A}^{\text{II}}_{ij} \dd u^i\wedge \dd u^j$ (here $i,j=1,\dots , 7$) 
we first check whether or not there is a subspace of such form invariant
for the action of $\stab(f_5)$, i.e.\ $\tilde A \longmapsto M_5^T \tilde A M_5=\tilde{A}$.
This is seen more easily considering the infinitesimal invariance condition,
i.e.\ by considering $M_5=I_{7}+m_5 + \ldots$, where $m_5\in\mathfrak{Lie}(\stab(\omega_5))$, the Lie algebra of the stabiliser. The infinitesimal condition
reads as $\tilde A \longmapsto m_5^T \tilde A+ \tilde A m_5=\mathbb{O}_{7,7}$, { where we indicate here and in what follows with $\mathbb{O}_{i,j}$ the $i\times j$ matrix with null coefficients}. This
gives a set of linear conditions for the coefficients of $\tilde{A}^{\text{II}}$ that admit
only the trivial solution. So, there is no invariant subspace and we can fix exactly 18 parameters of the 21 which are apriori free in $\tilde A^{\text{II}}$.

Then, by a direct computation with the Lie group $\stab(\omega_5)$ acting as above, i.e.\ as
$\tilde A \longmapsto M_5^T \tilde A M_5 $, we see that we can bring a generic 
$7\times 7$ skew-symmetric matrix $\tilde A$ in the following standard form:
\begin{equation}
    \tilde{A}^{\text{II}}=\begin{pmatrix}
        0&s_1&s_2&s_3&0&1&0\\
        -s_1&0&0&0&1&0&0\\
        -s_2&0&0&1&0&0&1\\
        -s_3&0&-1&0&0&0&0\\
        0&-1&0&0&0&0&0\\
        -1&0&0&0&0&0&1\\
        0&0&-1&0&0&-1&0
    \end{pmatrix},
\end{equation}
where $s_1,s_2,s_3$ are free parameters. In this case, 
\begin{align}\begin{split}
    A^{\text{II}}&= \dd u^1 \wedge (s_1  \dd u^2+ s_2 \dd u^3+s_3  \dd u^4 ) + \dd u^1\wedge \dd u^6+ \dd u^2\wedge \dd u^5+ \dd u^3\wedge \dd u^4 \\
    B^{\text{II}}&= (\dd u^3 + \dd u^6)\wedge \dd u^7  \end{split}
\end{align}

\subsection{Third orbit}
The leading coefficient of the operator associated to $\tilde T^{\text{III}}$ is 
\begin{equation}
    g_{ij}^{\text{III}}=\begin{pmatrix}
0&0&0&1&0&0\\
0&0&0&0&1&0\\
0&0&0&0&0&1\\
-1&0&0&0&u^6&-u^5\\
0&-1&0&-u^6&0&u^4\\
0&0&-1&u^5&-u^4&0
\end{pmatrix}
\end{equation} 
We remark that applying the linear transformation
\begin{equation}
    u^1\mapsto -u^5, \quad u^2\mapsto u^6, \quad u^3\mapsto u^7, \quad u^4\mapsto -u^2, \quad u^5\mapsto u^3 , \quad u^6\mapsto u^4, \quad  u^7\mapsto u^1
\end{equation} the three-form is mapped into 
\begin{equation}
    f_6= \dd u^1\wedge(\dd u^5\wedge \dd u^2+\dd u^7\wedge \dd u^4+\dd u^6\wedge \dd u^3)
    +\dd u^2\wedge \dd u^4\wedge \dd u^3
\end{equation}
which is case 3.6 in \cite{CohenHelminck1988}.

Following~\cite{CohenHelminck1988} we have the following structure of the stabiliser
$\stab (f_6)$ is composed of matrices $M_6=N_6L_6$, where the matrices
$N_6$ and $L_6$ have the following  shape\footnote{Here we correct two typos in the formula for the matrices
of the group $G_6:=\stab (f_6)$ presented in~\cite[\S 3.6]{CohenHelminck1988}: therein the coefficients $(5,3)$ and $(6,4)$ of the matrix $N_6$ have the incorrect sign.}:
\begin{subequations}
    \begin{align}
    N_6 &= 
    \begin{pmatrix}
        \lambda_1^{-1} & 0 & 0 & 0 & 0 & 0 & 0 
        \\
        \lambda_2 & 1 & 0 & 0 & 0 & 0 & 0
        \\
        \lambda_3 & 0 & 1 & 0 & 0 & 0 & 0
        \\
        \lambda_4 & 0 & 0 & 1 & 0 & 0 & 0
        \\
        \lambda_5 & 0 & -\lambda_4 & \lambda_3 & \lambda_1 & 0 & 0
        \\
        \lambda_6 & 0 & 0 & -\lambda_2 & 0 & \lambda_1 & 0
        \\
        \lambda_7 & 0 & 0 & 0 & 0 & 0 & \lambda_1 
            \end{pmatrix}
            \\[4pt]
    L_6 &=
    \begin{pmatrix}
        1 & \mathbb{O}_{1,3} & \mathbb{O}_{1,3} 
        \\
        \mathbb{O}_{3,1} & D & \mathbb{O}_{3,3} 
        \\
        \mathbb{O}_{3,1} & S & (D^{T})^{-1}
    \end{pmatrix}
    \end{align}
\end{subequations}
where $\lambda_i\in\K$, $D\in\SL(3,\K)$, and $S^T=S$. That is, we have 
$\dim\stab(f_6)=21$, and it can be shown that the group has structure 
$\stab(f_6) \cong \K^{12}\rtimes (\SL(3,\K)\times\K)$. 

Using the same argument as in the previous case,  through a direct computation with the Lie group $\stab(f_6)$ acting 
as above, i.e.\ 
$\hat A \longmapsto M_6^T \hat A M_6$ we see that there is no invariant subspace and we can map a generic 
$7\times 7$ skew-symmetric matrix $\tilde A$ in the following standard form:
\begin{equation}
    \hat{A}=\begin{pmatrix}
        0&0&1&0&0&0&1
        \\ 
        0&0&0&1&0&0&0
        \\ 
        -1&0&0&0&0&0&1
        \\
        0&-1&0&0&0&0&0
        \\
        0&0&0&0&0&1&0
        \\
        0&0&0&0&-1&0&0
        \\
        -1&0&-1&0&0&0&0
    \end{pmatrix},
\end{equation}
Returning to the original variables chosen in~\cite{VerVit2} through conjugation, 
we obtain:
\begin{equation}
    \tilde{A}^{\text{III}} =
    \begin{pmatrix}
        0& 0& 1& 0& 0& 0& 1
        \\
        0& 0& 0& 1& 0& 0& 0
        \\
        -1& 0& 0& 0& 0& 0& 1
        \\
        0& -1& 0& 0& 0& 0& 0
        \\
        0& 0& 0& 0& 0& 1& 0
        \\
        0& 0& 0& 0& -1& 0& 0
        \\
        -1& 0& -1& 0& 0& 0& 0
    \end{pmatrix},
\end{equation}
that is the following values of the forms $A$ and $B$:
\begin{align}\begin{split}
    A^{\text{III}}&= -\dd u^1 \wedge \dd u^2-\dd u^3\wedge \dd u^5-\dd u^4\wedge \dd u^6,
    \\
    B^{\text{III}}&=-\dd u^3-\dd u^5. 
    \end{split}
\end{align}
This finally determines the compatible class of hydrodynamic type systems.

\subsection{Forth orbit}

The leading coefficient of the operator associated to $\tilde T^{\text{IV}}$ is 
\begin{equation}
    g_{ij}^{\text{IV}}=\begin{pmatrix}
0&u^3&-u^2&1&0&0\\
-u^3&0&u^1&0&1&0\\
u^2&-u^1&0&0&0&0\\
-1&0&0&0&u^6&-u^5\\
0&-1&0&-u^6&0&u^4\\
0&0&0&u^5&-u^4&0
\end{pmatrix}\, .
\end{equation}

In this case, by applying the linear transformation
\begin{equation}
    u^1\mapsto u^4, \quad u^2\mapsto u^5, \quad u^3\mapsto u^2, \quad u^4\mapsto u^6, \quad u^5\mapsto u^7,   \quad u^6\mapsto u^3, \quad  u^7\mapsto u^1,
\end{equation}
the three-form is mapped into 
\begin{equation}
    f_7= \dd u^1 \wedge (\dd u^4 \wedge \dd u^6 + \dd u^5 \wedge \dd u^7) 
    + \dd u^2 \wedge \dd u^4 \wedge \dd u_5 
    + \dd u^3 \wedge \dd u^6 \wedge \dd u^7,
\end{equation}
which is case 3.7 in \cite{CohenHelminck1988} (see appendix \ref{app:forms}).

In this case we characterise the stabiliser $\stab(f_7)$ in a slightly
different way than it was done in~\cite{CohenHelminck1988}. Our approach will
yield a clearer insight on the structure of the Lie group and an easier 
decomposition which will simplify the subsequent computations. First of all,
we observe that it is easier to reconstruct the group $\stab(f_7)$ (or at least the
connected component  of identity, which is enough for our purposes), as
exponentiation of the corresponding Lie algebra $\mathfrak{Lie}(\stab(f_7))$.
For $M_7\in \stab(f_7)$ holds the invariance condition $M_7\cdot f_7 = f_7$
which infinitesimally reads as $\mathcal{L}_{m_7}f_7 = 0$, where
$\mathcal{L}_x$ is the Lie derivative in the direction of the vector
field $x$. This yields a set of linear equations, that can be readily
solved to give the following form of the matrix $m_7$:
\begin{equation}
    m_7 = \sum_{i=1}^{3}h_i H_i + \sum_{i=1}^{2}\left(x_iX_i+y_iY_i\right)
    +\sum_{i=1}^{8}b_iB_i.
\end{equation}
where
\begin{subequations}
    \begin{align}
        H_1 &= \diag( 0, -2, 2, 1, 1, -1, -1 ),
        \\
        H_2 &= \diag(0, 0, 0, -1, 1, 1, -1 ),
        \\
        H_3 &= \diag(-2, -2, -2, 1, 1, 1, 1 ),
        \\
        X_1 &=
        \begin{pmatrix}
            0&1&\multicolumn{5}{c}{\multirow{3}{*}{$\mathbb{O}_{3,5}$}}
            \\
            0&0&
            \\
            2&0&
            \\
            \multicolumn{5}{c}{\multirow{2}{*}{$\mathbb{O}_{2,5}$}}
            &0&1
            \\
            &&&&&-1&0
            \\
            \multicolumn{7}{c}{\mathbb{O}_{2,7}}
        \end{pmatrix},
        \\
        Y_1 &= 
        \begin{pmatrix}
            0&0&1&\multicolumn{4}{c}{\multirow{2}{*}{$\mathbb{O}_{2,4}$}}
            \\
            2&0&0&
            \\
            \multicolumn{7}{c}{\mathbb{O}_{3,7}}
            \\
            \multicolumn{3}{c}{\multirow{2}{*}{$\mathbb{O}_{2,3}$}}&0&-1&0&0
            \\
            &&&1&0&0&0
        \end{pmatrix},
        \\
        X_2 &= 
        \begin{pmatrix}
            \multicolumn{7}{c}{\mathbb{O}_{4,7}}
            \\
            \multicolumn{3}{c}{\multirow{3}{*}{$\mathbb{O}_{3,3}$}}&-1&0&0&0
            \\
            &&&0&0&0&1
            \\
            &&&0&0&0&0
        \end{pmatrix}, 
        \\
        Y_2 &=
        \begin{pmatrix}
            \multicolumn{7}{c}{\mathbb{O}_{3,7}}
            \\
            \multicolumn{4}{c}{\multirow{4}{*}{$\mathbb{O}_{4,4}$}}&-1&0&0
            \\
            &&&&0&0&0
            \\
            &&&&0&0&0
            \\
            &&&&0&1&0
        \end{pmatrix},
        \\
        \sum_{i=1}^{8}b_iB_i
    &=
    \begin {pmatrix} 
        \multicolumn{3}{c}{\multirow{3}{*}{$\mathbb{O}_{3,3}$}}& b_1& b_{2}& b_3& b_4
        \\ 
        \multicolumn{3}{c}{}&b_5 &b_6&b_{2}&-b_1
        \\
        \multicolumn{3}{c}{}& -b_4& b_3 & b_7 & b_8
        \\
        \multicolumn{3}{c}{\mathbb{O}_{4,3}}&\multicolumn{4}{c}{\mathbb{O}_{4,4}}
\end {pmatrix}.
    \end{align}
\end{subequations}
It is easy to see that $\mathfrak{b} := \langle B_1,\ldots,B_8\rangle$ form
an abelian subalgebra of $\mathfrak{Lie}(\stab(f_7))$. Moreover,
$\mathfrak{s}_i = \langle H_i,X_i,Y_i\rangle$, $i=1,2$ are two subalgebras
isomorphic to $\mathfrak{sl}(2,\K)$, and are such that:
\begin{equation}
    [\mathfrak{s}_1,\mathfrak{s}_2] = 
    [\mathfrak{s}_1,H_3] =
    [\mathfrak{s}_2,H_3] = 0,
    \quad
    [\mathfrak{s}_i,\mathfrak{b}] \subseteq \mathfrak{b},
    \quad
    [H_3,\mathfrak{b}] \subseteq \mathfrak{b},
\end{equation}
That is we have the following decomposition $\mathfrak{Lie}(\stab(f_7)) 
\cong (\mathfrak{sl}(2,\K)\oplus\mathfrak{sl}(2,\K)\oplus \K)\oplus_s \K^8$,
where $\oplus_s$ denotes the semidirect sum. In particular 
$\dim \mathfrak{Lie}(\stab(f_7)) =15$.

So, the group, or at least its connected 
component of the identity, which is enough for our purposes, is constructed by
exponentiation of the three subalgebras, i.e.\ $\stab(f_7)
\cong (\SL(2,\K)\times\SL(2,\K)\times \K)\rtimes \K^8$. More
explicitly $M_7\in \stab(\omega_7)$ decomposes as $M_7 = P_7Q_7R_7S_7$ where:
\begin{subequations}
    \begin{align}
        P_7 &=
        \begin{pmatrix}
            1&0&0&0&0&0&0
            \\
            0&1&0&0&0&0&0
            \\
            0&0&1&0&0&0&0
            \\
            0&0&0&\eta_2^{-1}&-\upsilon_2/\eta_2&0&0
            \\
            0&0&0&-\eta_2\xi_2 &\eta_2\left( \xi_2\upsilon_2+1 \right) &0&0
            \\
            0&0&0&0&0&\eta_2\left( \xi_2\upsilon_2+1\right) &\eta_2\xi_2
            \\
            0&0&0&0&0&\upsilon_2/\eta_2&\eta_2^{-1}
        \end{pmatrix}
        \\
        Q_7 &
        \begin{aligned}[t]
        &= \diag\left(
        1, \eta_1^{-2},\eta_1^{-2},\eta_1,\eta_1,\eta_1^{-1},\eta_1^{-1}
        \right)
        \\
        &\cdot
        \begin{pmatrix}
            1+2\xi_1\upsilon_1&\xi_1&\xi_1\upsilon_1(1+\upsilon_1)&0&0&0&0
            \\
            2\upsilon_1& 1 & \upsilon_1^{2}&0&0&0&0
            \\
            2\xi_1(1+\upsilon_1) & \xi_1^{2}&(1+\xi_1\upsilon_1)^2&0&0&0&0
            \\
            0&0&0&\xi_1\upsilon_1+1&0&0&\xi_1
            \\
            0&0&0&0&1+\xi_1\upsilon_1&-\xi_1&0
            \\
            0&0&0&0&-\upsilon_1&1&0
            \\
            0&0&0&\upsilon_1&0&0&1
        \end{pmatrix}   
        \end{aligned}
        \\
        R_7 &=
        \diag\left(
        \eta_3^{-2}, \eta_3^{-2},\eta_3^{-2},\eta_3,\eta_3,\eta_3,\eta_3
        \right)
        \\
        S_7 &=
        \begin{pmatrix}
        1&0&0&\beta_1&\beta_2&\beta_3&\beta_4
        \\
        0&1&0&\beta_5&\beta_6&\beta_2&-\beta_1
        \\
        0&0&1&-\beta_4&\beta_3&\beta_7&\beta_8
        \\
        0&0&0&1&0&0&0
        \\
        0&0&0&0&1&0&0
        \\
        0&0&0&0&0&1&0
        \\
        0&0&0&0&0&0&1
        \end{pmatrix},
    \end{align}
\end{subequations}
with parameters $\xi_i,\upsilon_i$, $i=1,2$, $\eta_i$, $i=1,2,3$, 
and $\beta_i$, $i=1,\ldots,8$. Clearly, like for the associated Lie algebra,
we have $\dim \stab(f_7) =15$.

We are now in place to find the a standard form for the the generic two-form 
$\tilde{A}=\tilde{A}_{ij} \dd u^i\wedge \dd u^j$ (here $i,j=1,\dots , 7$)
up to the action of $\stab(\omega_7)$. Like in the previous cases, we 
first check whether or not there is a subspace of such form invariant
for the action of $\stab(f_7)$ through the infinitesimal condition
$\tilde A \longmapsto m_7^T \tilde A\, + \tilde A m_7=\mathbb{O}_{7,7}$ with
$m_7\in\mathfrak{Lie}(\stab(f_7))$. Solving the associated conditions
we obtain that there is no invariant subspace and we can fix all the 15 parameters 
of $\tilde A$.

Again, through a direct computation with the Lie group $\stab(f_7)$ acting 
as above, i.e.\ 
$\tilde A \longmapsto M_7^T \tilde A M_7$ we see that we can bring a generic 
$7\times 7$ skew-symmetric matrix $\tilde A$ in the following standard form:
\begin{equation}
    \tilde{A}=
    \begin{pmatrix}
        0 & 0 & 0 & 0 & 1 & 0 & 1
        \\
        0 & 0 & s_{1} & 0 & 0 & 0 & 0
        \\
        0 & -s_{1} & 0 & 0 & 0 & 0 & 0
        \\
        0 & 0 & 0 & 0 & 1 & s_{2} & s_{3}
        \\
        -1 & 0 & 0 & -1 & 0 & s_{4} & s_{5}
        \\
        0 & 0 & 0 & -s_{2} & -s_{4} & 0 & s_{6}
        \\
        -1 & 0 & 0 & -s_{3} & -s_{5} & -s_{6} & 0
    \end{pmatrix}
\end{equation}
where $s_i$, $i=1,\ldots,6$, are free parameters. 
Returning to the original variables chosen in~\cite{VerVit2} through conjugation, 
we obtain:
\begin{equation}
    \tilde{A} =
    \begin{pmatrix}
        0&1&0&s_{{2}}&s_{{3}}&0&0
        \\
        -1&0&0&s_{{4}}&s_{{5}}&0&-1
        \\
        0&0&0&0&0&s_{{1}}&0
        \\
        -s_{{2}}&-s_{{4}}&0&0&s_{{6}}&0&0
        \\
        -s_{{3}}&-s_{{5}}&0&-s_{{6}}&0&0&-1
        \\
        0&0&-s_{{1}}&0&0&0&0
        \\
        0&1&0&0&1&0&0
    \end{pmatrix},
\end{equation}
that is the following values of the forms $A$ and $B$:
\begin{align}\begin{split}
    A^\text{IV}&
    \begin{aligned}[t]
        = \dd u^1 \wedge (\dd u^2 +s_2 \dd u^4+s_3 \dd u^5)
    &+ \dd u^2 \wedge (s_4\dd u^4+s_5\dd u^5)
    \\&+s_1 \dd u^3 \wedge \dd u^6 +s_6 \dd u^4 \wedge \dd u^5,
    \end{aligned}
    \\
    B^\text{IV}&=-\dd u^2-\dd u^5. 
    \end{split}
\end{align}
This finally determines the compatible class of hydrodynamic type systems.

\subsection{Fifth orbit: the  $\mathfrak{g}_2$ case}

Let us start from the open orbit. It gives the following leading coefficient of the operator
\begin{equation}
g_{ij}^{\text{V}}=\begin{pmatrix}
0&u^3&-u^2&1&0&0\\
-u^3&0&u^1&0&1&0\\
u^2&-u^1&0&0&0&1\\
-1&0&0&0&u^6&-u^5\\
0&-1&0&-u^6&0&u^4\\
0&0&-1&u^5&-u^4&0
\end{pmatrix}
\end{equation}
where $\text{det}(g^{\text V}_{ij})=(u^1u^4+u^2u^5+u^3u^6-1)^2$;

The three-form used in \cite{VerVit2} is
\begin{equation}
    \begin{aligned}
   \tilde{T} &=  \dd u^1 \wedge \dd u^2 \wedge \dd u^3 
   +\dd u^4 \wedge \dd u^5 \wedge \dd u^6
    \\
    &+\dd u^7 \wedge (\dd u^1 \wedge \dd u^4 +\dd u^2 \wedge \dd u^5 +\dd u^3 \wedge \dd u^6)     
    \end{aligned}
\end{equation}
and this is exactly  the representative $f_9$ from the classification in~\cite{CohenHelminck1988}, 
see also Appendix~\ref{app:forms}. However, according to \cite{Fulton_H} there is a 
convenient basis of the differentials $\{v^4,v^3,v^1,u,w^1,w^3,w^4\}$ to describe the action 
of $G_2$ on $\Lambda^2\K$. In this basis, the three-form reads as 
\begin{equation}
    \omega= w^3\wedge u\wedge v^3+v^4\wedge u\wedge w^4+w^1\wedge u\wedge v^1+2\left(v^1\wedge v^3\wedge w^4+w^1\wedge w^3\wedge v^4\right)
    \label{omeg_g2}
\end{equation}
where we used the transformation given by the following mapping
\begin{equation}
    \begin{gathered}
    \dd u^1\mapsto \sqrt[3]{2} \, v^1, \quad \dd u^2\mapsto \sqrt[3]{2}\, v^3, 
    \quad \dd u^3\mapsto \sqrt[3]{2}\, w^4, 
    \\
    \dd u^4\mapsto \sqrt[3]{2}\, w^1, \quad \dd u^5\mapsto \sqrt[3]{2}\, w^3, 
    \quad \dd u^6\mapsto \sqrt[3]{2}\, v^4, \quad \dd u^7\mapsto u/\sqrt[3]{4}
    \end{gathered}
\end{equation}
note that here the transformation is between the differentials and not for the local set of coordinates.
\vspace{3mm}

We first recall that  the vector space $\Lambda^3\KK^7$ has dimension $35$ and under the action
of $GL(7,\KK)$ there is an open orbit. 
The subgroup of $\SL(7,\KK)$ that fixes $\omega$ is a ($\KK$-form of) the Lie group of type $G_2$
of dimension $14$
(cf.\ \cite[Prop.\ 22.12]{Fulton_H}).
Using the usual decomposition in the new basis, $\omega$ is fixed by $G_2$, 
and so is $\dd u^8$ (that is unchanged by the introduced transformation), we only 
need to understand the action of $G_2$ on $\Lambda^2\KK^7$. The vector space 
$\Lambda^2\KK^7$ has dimension $\binom{7}{2}=21$.
This $G_2$-representation decomposes into irreducible $\lieg_2$-representations as (\cite[\S 22.3]{Fulton_H}):
\begin{equation}
    \Lambda^2\KK^7\,=\,\lieg_2\,\oplus\,\KK^7~,    
\end{equation}
where the two summands are, as $G_2$-representations, the adjoint representation and 
the standard representation respectively.
We remark that the Lie algebra $\lieg_2$ is a subalgebra of the algebra $M_7(\KK)$ of $7\times 7$ 
matrices and the action of $G_2$ on $\lieg$ is by conjugation: $g:X\mapsto gXg^{-1}$. The characteristic 
polynomial of $X$ is thus an invariant for this action. A general $X\in \lieg_2$ can be diagonalized by 
certain $g\in G_2$, so that the  eigenvalues of $X$ depend on two parameters (since $\lieg_2$ has rank 
two), hence there are two invariants for the action of $G_2$ on $\lieg_2$. The standard representation 
$\KK^7$ has a $G_2$-invariant quadratic form, in fact $G_2\subset SO(7)\subset GL(7,\KK)$, and one finds 
that this quadratic form provides the only $G_2$-invariant on $\KK^7$. Thus there are at least three 
$G_2$-invariants on $\Lambda^2\KK^7$. However, since $\dim \wedge^2\KK^7\,-\,\dim G_2=21-14=7$, there 
must be more invariants and one expects the general orbit to have dimension $14$, so that it depends on 
$7$ parameters.

Roughly speaking, we have to carry out the construction described in what follows. Given a general element $(X,v)\in \lieg_2\oplus\KK^7$, one can first use the $G_2$-action to diagonalize
$X$ and after this we denote the transformed element by $(H,w)$.
The stabilizer in $G_2$ of a general diagonal matrix is the subgroup of diagonal matrices, which has dimension two. These diagonal matrices act, in a suitable basis, as diagonal matrices also on $\KK^7$ and
thus for a general $w$ we may assume that two of the seven coordinates are equal to $1$. There remain $7-2=5$ free coordinates and adding the two parameters for the eigenvalues of $X$ we get $5+2=7$ parameters which
determine the orbit.

\vspace{3mm}

We make the decomposition of $\Lambda^2\K^7$ outlined in the previous paragraph explicit, using the basis $v^4,\dots w^4$ as in \eqref{omeg_g2}. The advantage of this basis is that the action of the Lie algebra $\lieg_2$ of $G_2$ on $\KK^7$ is
easy to describe. The $14$-dimensional Lie algebra $\mathfrak{g}_2$ is generated by 4 elements, $X_1,Y_1,X_2,Y_2$ and
their action on the chosen basis is given in \Cref{tab:action}, see 
\cite[p.\ 354]{Fulton_H}.

\begin{table}[ht!]
  \begin{center}
    \label{G2action}
    \begin{tabular}{cccccccc}
    \toprule
    &$v^4$&$v^3$&$v^1$&$u$&$w^1$&$w^3$&$w^4$\\
    \midrule
    $X_1$&$0$&$v^4$&$0$&$2v^1$&$u$&$0$&$-w^3$\\
    $Y_1$ &$v^3$ &$0$&$u$&$2w^1$&$0$&$-2w^4$&$0$\\
    $X_2$ &$0$ &$0$&$-v^3$&$0$&$0$&$w^1$&$0$\\
    $Y_2$ &$0$&$-v^1$&$0$&$0$&$w^3$&$0$&$0$\\
    \bottomrule
    \end{tabular}
  \end{center}
  \caption{Action of the generators $X_i$ and $Y_i$ on the basis of $\K^7$.}
 \label{tab:action}
\end{table}
The subalgebra of diagonal matrices (the Cartan subalgebra $\lieh$) in $\lieg_2$ is spanned by
\begin{equation}
H_i\,:=\,[X_i,Y_i],\quad (i=1,2),\qquad
\left\{\begin{array}{rcl}
        H_1&=&\mbox{diag}(1,-1,2,0,-2,1,-1),\\
        H_2&=&\mbox{diag}(0,1,-1,0,1,-1,0)~.
       \end{array}\right.    
\end{equation}
Thus the diagonal matrices in $G_2$ are $\mbox{diag}(\lambda,\lambda^{-1}\mu,\lambda^2\mu^{-1},1,
\lambda^{-2}\mu,\lambda\mu^{-1},\lambda^{-1})$. In particular, given a vector $(x_1,x_2,x_3,\ldots, x_7)$ with
$x_1,x_2\neq 0$, there is a diagonal matrix in $G_2$ that maps it to $(1,1,x_3,\ldots , x_7)$.

\vspace{3mm}


As outlined above, our goal is to identify the subspaces $\lieg_2$ and $\KK^7$ 
inside $\Lambda^2 \KK^7$. These subspaces are characterized as $\lieg_2$-subrepresentations. 
Consequently, the problem is reduced to studying the action of $\lieg_2$ on $\Lambda^2 \KK^7$ and 
comparing it with its action on the adjoint representation $\lieg_2$ and on the fundamental representation $\KK^7$.

The key ingredient is the Cartan subalgebra $\mathfrak h \subset \lieg_2$, 
namely the two-dimensional vector subspace spanned by $H_1$ and $H_2$. A fundamental 
property of $\mathfrak h$ is that its action is simultaneously diagonalisable on every 
finite-dimensional representation of $\lieg_2$. Thus, if $V$ is a representation of 
$\lieg_2$ and $v \in V$ is a common eigenvector for the action of $\mathfrak h$, then $v$ 
determines a linear functional
\begin{equation}
    \alpha_v \colon \mathfrak h \longrightarrow \KK,    
\end{equation}
defined by
\begin{equation}
    H \cdot v = \alpha_v(H)v \qquad \forall H \in \mathfrak h.    
\end{equation}
The functional $\alpha_v$ is called the \emph{weight} of $v$, see~\cite[p.\ 165]{Fulton_H}. The simple roots $\alpha_1,\alpha_2$ are weights of the adjoint representation and they are defined by
\begin{equation}
    \alpha_{1}(H_1)=\alpha_2(H_2)=2, \quad \alpha_1(H_2)=-3, \quad \alpha_2(H_1)=-1.
\end{equation} 

The strategy is therefore to compute the weight vectors of $\Lambda^2 \KK^7$. Since
$\Lambda^2 \K^7 \cong \lieg_2 \oplus \KK^7$, every weight vector of $\Lambda^2 \KK^7$ must 
belong either to the adjoint representation or to the 7-dimensional representation. 
In particular, if a weight occurs with multiplicity one in $\Lambda^2 \KK^7$, then 
its corresponding weight vector must lie entirely in one of the two summands. 
Indeed, if the same weight appeared in both summands, its multiplicity in $\Lambda^2 K^7$ 
would be at least two.

This observation allows us to identify explicit vectors belonging to either $\lieg_2$ or $\KK^7$.
Finally, since both are irreducible $\lieg_2$-modules, the action of $\lieg_2$ on any 
nonzero vector in the corresponding summand generates the entire subrepresentation. 
In this way, the full decomposition of $\Lambda^2 K^7$ into the two irreducible 
components is recovered.

The Lie algebra action on $\KK^7$ also induces actions
on $S^d\KK^7$, the polynomials, homogeneous of degree $d$, in the basis elements 
and on $\Lambda^k\KK^7$, the alternating forms, by derivations, so by 
the Leibniz rule. 
The following quadratic polynomial $q$ is an invariant for the action on $S^2\KK^7$, that is
$X\cdot q=0$ for all $X\in\lieg_2$:
\begin{equation}
q\,:=\,4(v^1w^1\,+\,v^3w^3\,+\,v^4w^4)\,-\,u^2~.    
\end{equation}
It suffices to check this for the four generators, for example:
\begin{equation}
X_1\cdot q\,=\,4((0+v^1u)\,+\,(v^4w^3+0)\,+\,(0-v^4w^3))\,-\,4uv^1\,=\,0~.    
\end{equation}
This implies that the representations of $\lieg_2$ on $\KK^7$ and its dual vector 
space $(\KK^*)^7$ are equivalent.

\vspace{3mm}

\begin{figure}
    \centering
  \begin{tikzpicture}[scale=1.2]

\foreach \ang/\lab in {
90/{\alpha_2},
30/{3\alpha_1+2\alpha_2},
330/{3\alpha_1+\alpha_2},
270/{-\alpha_2},
210/{-(3\alpha_1+2\alpha_2)},
150/{-(3\alpha_1+\alpha_2)}
}{
  \draw[->,blue!80!black,thick] (0,0)--(\ang:3cm);
  \node[scale=0.65] at (\ang:3.45cm) {$\lab$};
}

\foreach \ang/\lab in {
60/{\alpha_1},
120/{\alpha_1+\alpha_2},
180/{2\alpha_1+\alpha_2},
240/{-\alpha_1},
300/{-(\alpha_1+\alpha_2)},
360/{-(2\alpha_1+\alpha_2)}
}{
  \draw[->,blue!80!black,thick] (0,0)--(\ang:2cm);
  \node[scale=0.65] at (\ang:2.35cm) {$\lab$};
}

\draw[magenta,->] (1,0)
  arc(0:150:1cm)
  node[pos=0.1,right,scale=0.5] {$5\pi/6$};


\end{tikzpicture}
    \caption{The root system of the Lie algebra $\lieg_2$.}
    \label{fig:G2}
\end{figure}

With reference to \Cref{fig:G2}, we have that the highest weight of 
the adjoint representation is denoted by $\alpha_6=3\alpha_1+2\alpha_2$
(\cite[\S 22.1]{Fulton_H}). The root $\alpha_2$ has multiplicity one 
in $\lieg_2$ and in $\Lambda^2\K^7$, the corresponding
weight spaces are spanned by $X_2$ and $v^3\wedge w^1$, hence these elements correspond, up to scalar multiple,
in $\lieg_2\subset \Lambda^2\K^7$.
Since $\ad(Y_2)(X_2)=-H_2$ we find that $Y_2\cdot (v^3\wedge w^1)$ corresponds to the diagonal matrix $-H_2$ in $\lieg_2$.
Explicitly,
\begin{equation}
    Y_2\cdot (v^3\wedge w^1)\,=\,-v^1\wedge w^1\,+\,v^3\wedge w^3~.    
\end{equation}
{Similarly, using the (long) root $\alpha_5$ we find that $X_5$ and $v^4\wedge v^1$ correspond up to scalar multiple. Since $Y_3:=[Y_1,Y_2]$, $Y_4:=[Y_1,Y_3]$ and $Y_5:=[Y_1,Y_4]$ (cf.\ \cite[\S 22.1]{Fulton_H})
one computes that on the basis of $\K^7$ the action of $Y_5$ is $v^4\mapsto 6w^1$, $v^1\mapsto -6w^4$ and the
other basis vectors map to zero. Then $\ad(Y_5)(X_5)=-H_5$, a diagonal matrix, corresponds up to scalar multiple with
\begin{equation}
Y_5\cdot (v^4\wedge v^1)\,=\,6w^1\wedge v^1-6v^4\wedge w^4~.    
\end{equation}
This suffices to show that the image of the 2-dimensional Cartan subalgebra 
$\Psi:\lieh\hookrightarrow \Lambda^2\K^7$
consists of the following alternating $7\times 7$ matrices
(we omit the coefficients that are zero, except for one on the diagonal):
\begin{equation}
\Psi(h)=M_h=\begin{pmatrix} &&&&&&a\\&&&&&b&\\&&&&c&&\\&&&0&&&\\&&-c&&&\\&-b&&&&\\-a&&&&&&\end{pmatrix}\quad \mbox{with}\quad
a+b+c\,=\,0~.    
\end{equation}


To find the image of $\K^7\subset \wedge^2\K^7$, we observe that $w^4\in \K^7$ lies in the (lowest) weight space with weight $\beta_4$ and that $Y_1w^4=0=Y_2w^4$.
The weight space $(\wedge^2\K^7)_{\beta_4}$ is two dimensional with basis $u\wedge w^4$, $w^1\wedge w^3$.
The elements mapped to zero by both $Y_1$ and $Y_2$ are the scalar multiples of $u\wedge w^4+2w^1\wedge w^3$.
So we define $\Phi:\K^7\hookrightarrow \wedge^2\K^7$ by 
$w^4\mapsto u\wedge w^4+2w^1\wedge w^3$ and by
imposing $\lieg_2$-equivariance. Then one finds:
}

The map of $\lieg_2$-representions $V\hookrightarrow\wedge^2V$ is then given by:
$$
\Phi:\left\{\begin{array}{rcl}
v^4&\longmapsto&v^4\wedge u-2v^3\wedge v^1,\\
v^3&\longmapsto&-v^3\wedge u + 2v^4\wedge w^1,\\
v^1&\longmapsto&-v^1\wedge u - 2 v^4\wedge w^3,\\
u&\longmapsto&-2(v^1\wedge w^1 +v^3\wedge w^3 -v^4\wedge w^4),\\
w^1&\longmapsto&-u\wedge w^1 + 2v^3\wedge w^4,\\
w^3&\longmapsto&-u\wedge w^3 - 2v^1\wedge w^4,\\
w^4&\longmapsto&u\wedge w^4+2w^1\wedge w^3.
       \end{array}
\right.
$$
The alternating $7\times 7$ matrix $M_x$ is obtained after the application of the $\lieg_2$-equivariant map $\Phi:\K^7\hookrightarrow \Lambda^2\K^7$ to $x=x_1v^4+x_2v^3+x_3v^1+x_4u+x_5w^1+x_6w^3+x_7w^4\in \K^7$:
$$
\Phi(x)=M_x\,=\,\begin{pmatrix}
&&&x_1&2x_2&-2x_3&2x_4\\
&&-2x_1&-x_2&&-2x_4&2x_5\\
&2x_1&&-x_3&-2x_4&&-2x_6\\
-x_1&x_2&x_3&&-x_5&-x_6&x_7\\
-2x_2&&2x_4&x_5&&2x_7&\\
2x_3&2x_4&&x_6&-2x_7&&\\
-2x_4&-2x_5&2x_6&-x_7&&&\\
        \end{pmatrix}
$$
So, given a general element in $\wedge^2\K^7$, using the action of $G_2$ it can brought in form of an alternating matrix in $\Psi(\lieh)$ (depending on two parameters) summed with an $M_x$, where we may moreover assume that $x_1=x_2=1$.

We finally consider 
\begin{equation}
    A_{ij}=(M_h)_{ij}+(M_x)_{ij}, \qquad i,j=1,2,\dots 6
\end{equation}
where $x_1=x_2=1$:
\begin{equation}
  A=  \begin{pmatrix}
        0&0&0&1&2&-2x_3\\
        0&0&-2&-1&0&-2x_4+b\\
        0&2&0&-x_3&-a-b-2x_4&0\\
        -1&1&x_3&0&-x_5&-x_6\\
        -2&0&a+b+2x_4&x_5&0&2x_7\\
        2x_3&-b+2x_4&0&x_6&-2x_7&0
    \end{pmatrix}
\end{equation}
where $a,b,x_3,x_4,x_5,x_6$ and $x_7$ are arbitrary parameters (here $c=-a-b$ in $M_h$).
Now, $B_i=((M_\lieh)_{i7}+(M_x)_{i,7})$, so that 
\begin{equation}
   B=\begin{pmatrix}
        a+2x_4,& 2x_5,&-2x_6,&x_7,&0,&0
    \end{pmatrix}
\end{equation}
We stress that the pair obtained  $(A,B)$ is in the new coordinates according to \cite{Fulton_H}. We now need to invert the isomorphism between the basis to finally re-construct the representatives of this class in the original coordinates $u^1,\dots, u^7$. In the original ones we have
\begin{equation}
    \tilde{A}=\begin{pmatrix}
0 & 2^{\frac{1}{3}} & s_{6} 2^{\frac{1}{3}} & -\frac{\left(a +b +2 s_{4}\right) 2^{\frac{1}{3}}}{2} & 0 & 0 & -s_{3} 2^{\frac{1}{3}} 
\\
 -2^{\frac{1}{3}} & 0 & -s_{5} 2^{\frac{1}{3}} & 0 & \frac{\left(-2 s_{4}+b \right) 2^{\frac{1}{3}}}{2} & 0 & -2^{\frac{1}{3}} 
\\
 -s_{6} 2^{\frac{1}{3}} & s_{5} 2^{\frac{1}{3}} & 0 & 0 & 0 & \frac{\left(a +2 s_{4}\right) 2^{\frac{1}{3}}}{2} & s_{7} 2^{\frac{1}{3}} 
\\
 \frac{\left(a +b +2 s_{4}\right) 2^{\frac{1}{3}}}{2} & 0 & 0 & 0 & s_{7} 2^{\frac{1}{3}} & -2^{\frac{1}{3}} & s_{5} 2^{\frac{1}{3}} 
\\
 0 & -\frac{\left(-2 s_{4}+b \right) 2^{\frac{1}{3}}}{2} & 0 & -s_{7} 2^{\frac{1}{3}} & 0 & s_{3} 2^{\frac{1}{3}} & s_{6} 2^{\frac{1}{3}} 
\\
 0 & 0 & -\frac{\left(a +2 s_{4}\right) 2^{\frac{1}{3}}}{2} & 2^{\frac{1}{3}} & -s_{3} 2^{\frac{1}{3}} & 0 & 2^{\frac{1}{3}} 
\\
 s_{3} 2^{\frac{1}{3}} & 2^{\frac{1}{3}} & -s_{7} 2^{\frac{1}{3}} & -s_{5} 2^{\frac{1}{3}} & -s_{6} 2^{\frac{1}{3}} & -2^{\frac{1}{3}} & 0 
\end{pmatrix}
\end{equation}
so that 
\begin{subequations}
\begin{align}\begin{split}
    A=&
    2^{\frac13}\dd u^1\wedge \dd u^2
    +s_6\,2^{\frac13}\,\dd u^1\wedge \dd u^3
    -\frac{(a+b+2s_4)\,2^{\frac13}}{2}\,\dd u^1\wedge \dd u^4
    \\
    &-s_5\,2^{\frac13}\,\dd u^2\wedge \dd u^3
    +\frac{(-2s_4+b)\,2^{\frac13}}{2}\,\dd u^2\wedge \dd u^5
+\frac{(a+2s_4)\,2^{\frac13}}{2}\,\dd u^3\wedge \dd u^6
    \\
    &+s_7\,2^{\frac13}\,\dd u^4\wedge \dd u^5
-2^{\frac13}\,\dd u^4\wedge \dd u^6
+s_3\,2^{\frac13}\,\dd u^5\wedge \dd u^6.
\end{split}
\\
    B=&
-s_3\,2^{\frac13}\,\dd u^1
-2^{\frac13}\,\dd u^2
+s_7\,2^{\frac13}\,\dd u^3
+s_5\,2^{\frac13}\,\dd u^4
+s_6\,2^{\frac13}\,\dd u^5
+2^{\frac13}\,\dd u^6.
\end{align}    
\end{subequations}
The resulting system is highly nonlinear.

\section{Conclusions}
\label{sec:concl}

The classification presented here provides a complete list of
operator–system pairs with six components for which a second-order
Hamiltonian structure can be described. The results offer a physical
characterisation of evolutionary systems of conservation laws and, in
principle, may arise in concrete phenomena and models. The extensive use of
algebraic structures and geometric methods also highlights a more explicit
connection between mathematical physics and algebro-geometric approaches,
offering a concrete framework for the analysis of Hamiltonian systems.  At
present, analogous computations for systems with a higher number of
components are not feasible. However, we anticipate that a future
classification of forms, the corresponding group actions, and their
stabilisers will be developed by experts in the field.

\subsection*{Acknowledgments}

We thank Dr Danilo Latini for the helpful discussion during the preparation of
this paper.

GG, and PV  acknowledge the support of the research project Mathematical
Methods in NonLinear Physics (MMNLP), Gruppo 4-Fisica Teorica of INFN of the Sections of Milano and Lecce respectively.  This
work has been partially supported by the National Group of Mathematical
Physics (GNFM) of the Italian Institute for High Mathematics (INdAM).


\appendix

\section{Complete list of alternating three-forms for $n=7$ and their stabilisers}
\label{app:forms}

For sake of completeness here we report the complete result
of~\cite{CohenHelminck1988}. We have that any nonzero alternating three-form on
$\KK^{7}$, with $\KK=\mathbb{R},\mathbb{C}$, is equivalent through the action
of $\GL(7,\KK)$ to one of the following nine forms:
\begin{align}
    f_{1} &= \dd u^{1} \wedge \dd u^{2}\wedge \dd u^{3},
    \label{eq:I}
    \\
    f_{2} &= \dd u^{1} \wedge \dd u^{2}\wedge \dd u^{3}
    +\dd u^{1} \wedge \dd u^{4} \wedge \dd u^{5},
    \label{eq:II}
    \\
    f_{3} &= \dd u^{1} \wedge \dd u^{2}\wedge \dd u^{3}
    +\dd u^{4} \wedge \dd u^{5} \wedge \dd u^{6},
    \label{eq:III}
    \\
    f_{4} &= \dd u^{1} \wedge \dd u^{6} \wedge \dd u^{2}
    +\dd u^{2} \wedge \dd u^{4} \wedge \dd u^{3}
    +\dd u^{1} \wedge \dd u^{3} \wedge \dd u^{5},
    \label{eq:IV}
    \\
    f_{5} &= \dd u^{1} \wedge \dd u^{2}\wedge \dd u^{3}
    +\dd u^{4} \wedge \dd u^{5} \wedge \dd u^{6}
    +\dd u^{1} \wedge \dd u^{4} \wedge \dd u^{7},
    \label{eq:V}
    \\
    f_{6} &=
    \begin{aligned}[t]
         \dd u^{1} \wedge \dd u^{5}\wedge \dd u^{2}
        &+\dd u^{1} \wedge \dd u^{7} \wedge \dd u^{4}
        \\
        &+\dd u^{1} \wedge \dd u^{6} \wedge \dd u^{3}
        +\dd u^{2} \wedge \dd u^{4} \wedge \dd u^{3},
    \end{aligned}
    \label{eq:VI}
    \\
    f_{7} &= 
    \begin{aligned}[t]
        \dd u^{1} \wedge \dd u^{4} \wedge\dd u^{6}
        &+\dd u^{1} \wedge \dd u^{5} \wedge \dd u^{7}
        \\
        &+\dd u^{2} \wedge \dd u^{4} \wedge \dd u^{5}
        +\dd u^{3} \wedge \dd u^{6} \wedge \dd u^{7},
    \end{aligned}
    \label{eq:VII}
    \\
    f_{8} &= \dd u^{1} \wedge \dd u^{2}\wedge \dd u^{3}
    +\dd u^{1} \wedge \dd u^{4} \wedge \dd u^{5}
    +\dd u^{1} \wedge \dd u^{6} \wedge \dd u^{7},
    \label{eq:VIII}
    \\
    f_{9} &= 
    \begin{aligned}[t]
    \dd u^{1} \wedge \dd u^{2} \wedge\dd u^{3}
    &+\dd u^{4} \wedge \dd u^{5} \wedge \dd u^{6}
    +\dd u^{1} \wedge \dd u^{4} \wedge \dd u^{7}
    \\
    &+\dd u^{2} \wedge \dd u^{5} \wedge \dd u^{7}
    +\dd u^{3} \wedge \dd u^{6} \wedge \dd u^{7}.
    \end{aligned}
    \label{eq:IX}
\end{align}
Moreover, each alternating three form admits a non-trivial stabiliser in
$\GL(7,\KK)$ whose structure is reported in \Cref{tab:stabs}. In
Section~\ref{sec:class} when needed we corrected some misprints in the explicit
expression of the stabilisers present in the original
text~\cite{CohenHelminck1988}. However, we observe that the global structure of
the groups is the same as found in~\cite{CohenHelminck1988}.

\begin{table}[hbt]
    \centering
    \begin{tabular}{cc}
        \toprule
        Form & Stabiliser group
        \\
        \midrule
        $f_{1}$ & $\KK^{12}\rtimes\left( \SL \right(3,\KK)\rtimes\GL(4,\KK))$
        \\
        $f_{2}$ & $\KK^{14}\rtimes\left( \GL \right(2,\KK)\rtimes(\Sp(4,\KK)\rtimes\KK^{*}))$
        \\
        $f_{3}$ & $(\KK^{6}\rtimes\left( \SL \right(3,\KK)\rtimes(\SL(3,\KK)\rtimes\KK^{*})))\rtimes\mathbb{Z}_{2}$
        \\
        $f_{4}$& $\KK^{14}\rtimes\left( \GL \right(3,\KK)\rtimes\KK^{*})$
        \\
        $f_{5}$ & $(\KK^{12}\rtimes\left( \GL \right(2,\KK)\rtimes \GL(2,\KK)))\rtimes\mathbb{Z}_{2}$
        \\
        $f_{6}$ & $\KK^{14}\rtimes\left( \GL \right(3,\KK)\rtimes\KK^{*})$
        \\
        $f_{7}$ & $\KK^{8}\rtimes\left( (\GL \right(2,\KK)\rtimes\GL(2,\KK))/\KK^{*})$
        \\
        $f_{8}$ & $\KK^{6}\rtimes(\Sp(6,\KK)\rtimes\KK^{*})$
        \\
        $f_{9}$ & $G_{2}$
        \\
        \bottomrule
    \end{tabular}
    \caption{Stabiliser groups for the three-forms $f_{1}$, \dots, $f_{9}$.}
    \label{tab:stabs}
\end{table}

We observe that $\omega_{i}$ for $i=1,2,3,4$ correspond to degenerate operators
and are not of interest. On the other hand, we have that $\omega_5$ corresponds
to $g^2$, $\omega_6$ corresponds to $g^4$, $\omega_7$ corresponds to $g^3$,
$\omega_8$ corresponds to $g^1$ and finally $\omega_{9}$ corresponds to $g^5$.
The explicit linear transformations which maps the three-forms used
in~\cite{VerVit2} to obtain the leading coefficients $g_i$ are discussed when
needed in Section~\ref{sec:class}.

\end{document}